\documentclass[letterpaper, 10pt, conference]{ieeeconf}  

\IEEEoverridecommandlockouts                              

\usepackage{amsmath}
\usepackage{amssymb}
\usepackage{amsfonts}
\usepackage{mathptmx} 
\usepackage{times}
\usepackage{graphicx}
\DeclareMathAlphabet{\bit}{OML}{cmm}{b}{it}

\newtheorem{thm}{Theorem}

\usepackage{datetime}
\usepackage{framed} 
\usepackage{graphicx} 
\usepackage{amsmath} 
\usepackage{amssymb}  

\def\<{\leqslant}           
\def\>{\geqslant}           

\def\d{\partial}

\def\Re{\mathrm{Re}}   

\def\mZ{\mathbb{Z}}    
\def\mR{\mathbb{R}}    
\def\mC{\mathbb{C}}    

\def\rT{\mathrm{T}}        

\def\[[[{[\![\![}
\def\]]]{]\!]\!]}

\def\bra{{\langle}}
\def\ket{{\rangle}}

\def\dbra{\langle\!\!\langle}
\def\dket{\rangle\!\!\rangle}

\def\re{\mathrm{e}}        
\def\rd{\mathrm{d}}        

\def\br{\mathbf{r}}
\def\x{\times}
\def\ox{\otimes}

\def\fS{\mathfrak{S}}

\def\sK{\mathsf{K}}
\def\sL{\mathsf{L}}

\def\sH{\mathsf{H}}
\def\sA{\mathsf{A}}

\def\sB{\mathsf{B}}

\def\sC{\mathsf{C}}

\def\sD{\mathsf{D}}

\def\sX{\mathsf{X}}

\def\sY{\mathsf{Y}}

\def\sU{\mathsf{U}}

\def\sV{\mathsf{V}}

\def\sG{\mathsf{G}}

\def\fX{\mathfrak{X}}
\def\fU{\mathfrak{U}}

\def\sB{{\sf B}}
\def\sC{{\sf C}}

\def\fY{\mathfrak{Y}}

\def\mT{\mathbb{T}}
\def\mZ{\mathbb{Z}}

\def\sn{|\!|\!|}

\def\esssup{\mathop{\mathrm{ess\, sup}}}    
\def\essinf{\mathop{\mathrm{ess\, inf}}}    

\title{\LARGE \bf
Spatio-temporal Transfer Function Conditions of Positive Realness for Translation Invariant Lattice Networks of Interacting Linear
Systems$^*$}

\usepackage{datetime}

\author{
Igor G. Vladimirov$^{\dagger}$,
\qquad
Ian R. Petersen$^{\dagger}$\\%
\thanks{$^*$This work is supported by the Australian Research Council under grant DP160101121.}
\thanks{$^\dagger$Research School of Engineering, College of Engineering and Computer Science, Australian National University, ACT 2601, Canberra, Australia.
{\tt igor.g.vladimirov@gmail.com}, {\tt i.r.petersen@gmail.com}.
}
}

\begin{document}
\maketitle
\thispagestyle{empty}

\begin{abstract}
This paper is concerned with networks of interacting linear systems at sites of a multidimensional lattice. The systems are governed by linear ODEs with constant coefficients driven by external inputs, and their internal dynamics and coupling with the other component systems are translation invariant. Such systems occur, for example, in  finite-difference models of large-scale flexible structures manufactured from homogeneous materials. Using the spatio-temporal transfer function of this translation invariant network, we establish conditions for its positive realness in the sense of energy dissipation. The latter is formulated in terms of block Toeplitz bilinear forms of the input and output variables of the composite system. We also discuss quadratic stability of the network in isolation from the environment
and phonon theoretic dispersion relations.
%
\end{abstract}




\section{INTRODUCTION}

Complex physical systems can be viewed as a large number of relatively simple subsystems whose collective behaviour is a cumulative effect of their interaction rather than a particular individual  structure. Spatially homogeneous states of matter are modelled as identical building blocks which interact with each other in a translationally invariant fashion. A natural example of large-scale composite systems with translational symmetry is provided by crystalline solids, where spatially periodic arrangements of constituent particles result from their interaction and play an important role in their thermodynamic and mechanical properties (including the heat transfer and wave propagation) studied in the phonon theory  \cite{S_1990}.

Modern engineering exploits
translation  invariant interconnections  in drone swarming, vehicle  platooning and artificially fabricated metamaterials   \cite{VBSH_2006}, such as split ring resonator arrays with unusual electrodynamic characteristics (a negative refraction index). Nontrivial input-output properties of such networks of systems (natural or artificial) are not merely a ``sum'' of individual internal dynamics of their constituent blocks  and come from a specific structure of energy flows through the translation invariant interaction.

The energy balance relations,  which reflect the energy conservation and dissipation in isolated and open systems (for example, due to electrical resistance and mechanical friction), significantly affect the behaviour
of physical systems and play an increasingly important role in control design \cite{OVMM_2001,OVME_2002,VJ_2014}. These equations involve the internal energy and the work done on the system (which are represented in the dissipativity theory \cite{W_1972} in terms of storage and  supply rate functions). Work is modelled by using a bilinear form of the input and output variables, which are interpreted as the generalised force and velocity  respectively. For linear time-invariant systems with a finite-dimensional internal state,  the properties of being passive, positive real or negative imaginary (in the case of position variables instead of the velocity as the output) admit criteria in the form of linear matrix inequalities for the transfer functions in the frequency domain or the state-space matrices themselves \cite{PL_2010,XPL_2010}.

The present paper is concerned with similar conditions for networks of interacting linear systems at sites of a multidimensional lattice. The composite system is governed by an infinite set of linear ODEs with constant coefficients driven by external inputs, and their internal dynamics and coupling with the other component systems are translation invariant. These ODEs have block Toeplitz state-space matrices and can be represented in the spatio-temporal frequency domain by using appropriately modified transfer functions of several variables.
Such systems arise, for example, as finite-difference approximations of PDEs for large-scale flexible structures made of spatially homogeneous materials. Using the spatio-temporal transfer function of this translation invariant network, we establish conditions for its positive realness in the sense of energy dissipation. The latter is formulated in terms of block Toeplitz bilinear forms of the input and output variables of the composite system. The multivariate  Laplace and Fourier  transform techniques, which are used for this purpose, are similar to those for distributed control systems in the classical and quantum settings \cite{SVP_2015,VP_2014}.

The paper is organised as follows.
Section~\ref{sec:sys} describes the class of translation invariant networks under consideration.
Section~\ref{sec:freq} represents the network dynamics in terms of the spatial Fourier transforms of its signals.
Section~\ref{sec:bal} discusses energy balance relations in the case of bilinear supply rate and quadratic storage functions.
Section~\ref{sec:pass} establishes conditions for passivity of the network in terms of its spatio-temporal transfer function, and also discusses quadratic stability bounds for a dissipative network in isolation from the environment.
Section~\ref{sec:phon} considers phonon theoretic dispersion relations for the isolated network.
Section~\ref{sec:conc} provides concluding remarks.

\section{TRANSLATION INVARIANT NETWORKS
}\label{sec:sys}

We consider a translation invariant network of coupled linear systems at sites of the $\nu$-dimensional integer lattice $\mZ^{\nu}$. For any spatial index $k \in \mZ^{\nu}$, the $k$th system is endowed with an $\mR^n$-valued  vector $x_k$ of time-varying state variables (for example, the generalised positions and velocities). Associated with the $k$th lattice site are vectors $u_k$ and $y_k$  of external input and output variables which take values in $\mR^m$ and $\mR^r$, respectively, and also vary in time (unless indicated otherwise, vectors are organised as columns).
The states and outputs of these systems are governed by an infinite  set of coupled ODEs
\begin{align}
\label{xj}
  \dot{x}_j
  & =
  \sum_{k \in \mZ^\nu}
  (A_{j-k} x_k + B_{j-k} u_k),\\
\label{yj}
  y_j
  & =
  \sum_{k\in \mZ^\nu}
  (C_{j-k} x_k + D_{j-k} u_k)
\end{align}
for all $j\in \mZ^\nu$, where
$\dot{(\,)}$ is the time derivative (the time arguments are often omitted for brevity). Their right-hand sides are organised as convolutions, with the matrices $A_\ell \in \mR^{n\x n}$, $B_\ell\in \mR^{n\x m}$, $C_\ell\in \mR^{r\x n}$, $D_\ell \in \mR^{r\x m}$ depending on the relative location $\ell \in \mZ^\nu$ of the lattice sites in accordance with the translation invariance of the individual dynamics  of the systems and their coupling.
By assembling the inputs, states and outputs into the infinite-dimensional vectors $u:= (u_k)_{k \in \mZ^\nu}$, $x:= (x_k)_{k \in \mZ^\nu}$, $y:= (y_k)_{k \in \mZ^\nu}$, the set of ODEs (\ref{xj}) and (\ref{yj}) can be written as
\begin{align}
\label{x}
  \dot{x}
  & =
  A x + B u,\\
\label{y}
  y
  & =
  Cx+Du,
\end{align}
where the matrices $A:= (A_{j-k})_{j,k \in \mZ^\nu}$, $B:= (B_{j-k})_{j,k \in \mZ^\nu}$, $C:= (C_{j-k})_{j,k \in \mZ^\nu}$,  $D:= (D_{j-k})_{j,k \in \mZ^\nu}$ are block Toeplitz in the sense of the additive group structure of the lattice $\mZ^\nu$. The ODE (\ref{x}) is understood as a Volterra integral equation (of the second kind)
\begin{equation}
\label{xint}
    x(t)
    =
    \int_0^t
    (A x(\tau) + B u(\tau))
    \rd \tau
\end{equation}
whose solution is given by
\begin{equation}
\label{xsol}
    x(t)
    =
    \re^{tA} x(0) + \int_0^t \re^{(t-\tau) A} B u(\tau)\rd \tau,
    \qquad
    t\>0.
\end{equation}
For completeness, we note that (\ref{x})--(\ref{xsol}) can be endowed with a rigorous meaning as follows. To this end,
the state-space matrices in (\ref{xj}) and (\ref{yj}) are assumed to be square summable:
\begin{equation}
\label{ABCDgood}
    \sum_{\ell\in \mZ^\nu}
    \left\|
  \begin{bmatrix}
    A_\ell & B_\ell\\
    C_\ell & D_\ell
  \end{bmatrix}
    \right\|^2
    <+\infty,
\end{equation}
where $\|\cdot\|$ is an arbitrary matrix norm (whose particular choice is irrelevant in this case).  The fulfillment of (\ref{ABCDgood}) allows the Fourier series
\begin{equation}
\label{ABCD}
  \begin{bmatrix}
    \sA(\sigma) & \sB(\sigma)\\
    \sC(\sigma) & \sD(\sigma)
  \end{bmatrix}
  :=
  \sum_{\ell\in \mZ^\nu}
  \re^{-i\ell^{\rT}\sigma}
  \begin{bmatrix}
    A_\ell & B_\ell\\
    C_\ell & D_\ell
  \end{bmatrix},
\end{equation}
to be defined (in a blockwise fashion)
for almost all $\sigma \in \mR^\nu$ as the $L^2$-limit of appropriate partial sums over finite subsets of the lattice forming an exhausting sequence. Since the matrix-valued  functions $\sA$, $\sB$, $\sC$, $\sD$ (with values in $\mC^{n\x n}$, $\mC^{n\x m}$, $\mC^{r\x n}$, $\mC^{r\x m}$, respectively)
are $2\pi$-periodic with respect to each of their arguments, then, without loss of generality,  they can be considered on a $\nu$-dimensional torus $\mT^\nu$ (where $\mT$ can be  identified with the interval $[-\pi,\pi)$). These functions are Hermitian in the sense that $\overline{\sA(\sigma)} = \sA(-\sigma)$ (and similarly for the other functions $\sB$, $\sC$, $\sD$) for all $\sigma \in \mR^\nu$, where $\overline{(\cdot )}$ is the complex conjugate.  Furthermore, let
$
    \|u\|_2
    :=
    \sqrt{\sum_{k \in \mZ^\nu} |u_k|^2}$,
$
    \|x\|_2
    := \sqrt{\sum_{k \in \mZ^\nu} |x_k|^2}$,
$\|y\|_2
    :=
    \sqrt{\sum_{k \in \mZ^\nu} |y_k|^2}
$
denote the norms of the input, state and output of the network at a fixed but otherwise arbitrary moment of time $t\>0$ in the corresponding Hilbert spaces
\begin{equation}
\label{fU_fX_fY}
    \fU
    := L^2(\mZ^\nu, \mR^m),
    \
     \fX
     :=
     L^2(\mZ^\nu,\mR^n),
     \
     \fY
     :=
     L^2(\mZ^\nu, \mR^r)
\end{equation}
of appropriately dimensioned square summable real vector-valued functions $f:= (f_k)_{k \in \mZ^\nu}$ and $g:=(g_k)_{k \in \mZ^\nu}$ on the lattice $\mZ^\nu$ with the inner product $\bra f, g\ket := \sum_{k\in \mZ^\nu} f_k^\rT g_k$. Now, suppose the state-space matrices in (\ref{x}) and (\ref{y})  specify bounded linear operators $A: \fX \to \fX$, $B:\fU\to \fX$, $C:\fX\to \fY$, $D:\fU\to \fY$. This is equivalent to their $L^2$-induced operator norms being finite:
\begin{align}
\label{Anorm}
    \|A\|_{\infty}
    & =
    \esssup_{\sigma \in \mT^\nu}
    \|\sA(\sigma)\|
    <+\infty,\\
\label{Bnorm}
    \|B\|_{\infty}
    & =
    \esssup_{\sigma \in \mT^\nu}
    \|\sB(\sigma)\|
    <+\infty,\\
\label{Cnorm}
    \|C\|_{\infty}
    & =
    \esssup_{\sigma \in \mT^\nu}
    \|\sC(\sigma)\|
    <+\infty,\\
\label{Dnorm}
    \|D\|_{\infty}
    & =
    \esssup_{\sigma \in \mT^\nu}
    \|\sD(\sigma)\|
    <+\infty,
\end{align}
where the essential supremum is applied to the operator  norms of the matrices in (\ref{ABCD}).
In particular, if the matrices  $A_\ell$, $B_\ell$, $C_\ell$, $D_\ell$  vanish for all $\ell \in \mZ^\nu$ with $|\ell|$ large enough (so that each of the component systems in (\ref{xj}) and (\ref{yj}) is coupled with a finite number of the other systems in the network), then the functions $\sA$, $\sB$, $\sC$, $\sD$ in (\ref{ABCD}) are multivariate trigonometric polynomials, and the conditions (\ref{Anorm})--(\ref{Dnorm}) are satisfied in this case.  
In general, the fulfillment of (\ref{Anorm}) and (\ref{Bnorm}) guarantees that the operator exponential $\re^{tA}$ and the product $\re^{(t-\tau)A}B$ in (\ref{xsol}) are bounded block Toeplitz   operators acting on the Hilbert space $\fX$ and from $\fU$ to $\fX$, respectively, with
\begin{equation}
\label{bounds}
    \|\re^{tA}\|_{\infty} \< \re^{\|A\|_{\infty} t},
    \qquad
    \|\re^{(t-\tau)A} B\|_{\infty} \< \re^{\|A\|_{\infty}(t-\tau)}\|B\|_{\infty}
\end{equation}
for any $t\> \tau \> 0$, in view of
the submultiplicativity of the operator norm and the fact that block Toeplitz matrices form an algebra.
Therefore, if the initial network state is square summable, that is,
\begin{equation}
\label{x0norm}
    \|x(0)\|_2 < +\infty,
\end{equation}
and the network input is locally absolutely integrable with respect to time in the sense that
\begin{equation}
\label{ugood}
    \int_0^T
    \|u(t)\|_2
    \rd t
    <+\infty,
    \qquad
    T>0,
\end{equation}
then these properties are inherited by the subsequent states of the network. Indeed, a combination of (\ref{xsol}) with (\ref{bounds})--(\ref{ugood}) leads to
 \begin{equation}
\label{xgood}
    \|x(t)\|_2
    \<
    \re^{\|A\|_{\infty}t}
    \Big(
        \|x(0)\|_2
        +
    \|B\|_{\infty}
    \int_0^t
    \|u(\tau)\|_2
    \rd \tau
    \Big)
    <+\infty
\end{equation}
for all $t\>0$, and hence,
 \begin{align}
\nonumber
    \int_0^T
    \|x(t)\|_2
    \rd t
    \< &
    T
    \re^{\|A\|_{\infty}T}
    \Big(
        \|x(0)\|_2\\
\label{xgood1}
        & +
    \|B\|_{\infty}
    \int_0^T
    \|u(\tau)\|_2
    \rd \tau
    \Big)
    <+\infty
\end{align}
for any $T>0$.
Together with (\ref{y}),  (\ref{Cnorm}), (\ref{Dnorm}),  the property (\ref{xgood}) implies that the output norm $\|y(t)\|_2$ is finite for almost all $t>0$ since so is $\|u(t)\|_2$ in view of (\ref{ugood}). Furthermore, due to (\ref{xgood1}), the network output is also locally absolutely integrable with respect to time:
\begin{align}
\nonumber
    \int_0^T
    \|y(t)\|_2
    \rd t
    \< &
    \|C\|_{\infty}
    \int_0^T
    \|x(t)\|_2
    \rd t\\
\label{ygood}
    & +
    \|D\|_{\infty}
    \int_0^T
    \|u(t)\|_2
    \rd t
    <+\infty
\end{align}
for any $T>0$. In Section~\ref{sec:bal}, we will replace (\ref{ugood}) with a stronger condition on the network inputs in order to guarantee time-local square integrability for the outputs instead of (\ref{ygood}).

\section{NETWORK DYNAMICS IN THE SPATIAL FREQUENCY DOMAIN}\label{sec:freq}

The preservation of the spatial square summability (\ref{xgood}) (provided the input satisfies (\ref{ugood})) allows the network dynamics (\ref{x}) and (\ref{y}) to be represented in the spatial frequency domain as
\begin{align}
\label{X}
  \d_t X(t,\sigma)
  & =
  \sA(\sigma) X(t,\sigma) + \sB(\sigma) U(t,\sigma),\\
\label{Y}
  Y(t,\sigma)
  & =
  \sC(\sigma) X(t,\sigma)+\sD(\sigma) U(t,\sigma).
\end{align}
Here, similarly to (\ref{ABCD}), the functions $U$, $X$, $Y$ on $\mR_+\x \mT^\nu$ with values in $\mC^m$, $\mC^n$, $\mC^r$ are the Fourier transforms of the input, state and output of the network:
\begin{align}
\label{uDFT}
    U(t,\sigma)
    & :=
    \sum_{\ell\in \mZ^\nu}
    \re^{-i\ell^{\rT}\sigma}
    u_\ell(t),\\
\label{xDFT}
    X(t,\sigma)
    & :=
    \sum_{\ell\in \mZ^\nu}
    \re^{-i\ell^{\rT}\sigma}
    x_\ell(t),\\
\label{yDFT}
    Y(t,\sigma)
    & :=
    \sum_{\ell\in \mZ^\nu}
    \re^{-i\ell^{\rT}\sigma}
    y_\ell(t)
\end{align}
for almost all $\sigma \in \mT^\nu$ at time $t\> 0$. For any given $\sigma$, the equations (\ref{X}) and (\ref{Y}) describe an autonomous system (which is independent of the other systems in this parametric family with different values of $\sigma$) with a finite-dimensional internal state $X(\cdot, \sigma)$.  Their solution can be represented in terms of the Laplace transform over the time variable as
\begin{align}
\nonumber
  \sX(s,\sigma&)
   :=
  \int_{0}^{+\infty}
  \re^{-st}
  X(t,\sigma)
  \rd t\\
\label{cX}
    =&
    (sI_n - \sA(\sigma))^{-1} (X(0,\sigma)+ \sB(\sigma)\sU(s,\sigma)),\\
\nonumber
  \sY(s,\sigma&)
  :=
  \int_{0}^{+\infty}
  \re^{-st}
  Y(t,\sigma)
  \rd t\\
\nonumber
  = &
  \sC(\sigma) \sX(s,\sigma)+\sD(\sigma) \sU(s,\sigma)\\
\label{cY}
  = &
  \sC(\sigma)
  (sI_n - \sA(\sigma))^{-1} X(0,\sigma)
  +
  F(s,\sigma)\sU(s,\sigma),
\end{align}
where
\begin{equation}
\label{cU}
  \sU(s,\sigma)
   :=
  \int_{0}^{+\infty}
  \re^{-st}
  U(t,\sigma)
  \rd t.
\end{equation}
Here,
\begin{equation}
\label{F}
  F(s,\sigma)
  :=
  \sC(\sigma)(sI_n - \sA(\sigma))^{-1} \sB(\sigma) + \sD(\sigma)
\end{equation}
is the spatio-temporal transfer function of the network with values in $\mC^{r\x m}$. In view of (\ref{Anorm}), the relations (\ref{cX})--(\ref{F}) are valid for all $s\in \mC$ satisfying
\begin{equation}
\label{sgood}
  \Re s
  >
  \max
  \Big(
    \esssup_{\sigma \in \mT^\nu}
    \ln \br(\re^{\sA(\sigma)}),\,
  \limsup_{t\to +\infty}
    \frac{    \ln
    \|u(t)\|_2}{t}
  \Big),
\end{equation}
where $\br(\cdot)$ denotes the spectral radius of a square matrix (so that $\ln\br(\re^N) = \max \Re\fS(N)$
is the largest real part of the eigenvalues of $N$,
with $\fS(N)$ denoting its spectrum).
The presence of the upper limit in (\ref{sgood}) makes the integral in (\ref{cU}) convergent in the Hilbert space $L^2(\mT^\nu, \mC^m)$ due to the Parseval identity
$
    \|U(t,\cdot)\|_2
     :=
    \sqrt{
    \int_{\mT^\nu}
    |U(t,\sigma)|^2
    \rd \sigma}
    =
    (2\pi)^{\nu/2}
    \|u(t)\|_2
$
for any $t\>0$. In the case when the input vanishes, so that the network is effectively isolated from the environment,   (\ref{X}) reduces to a homogeneous linear ODE
$
    \d_t X(t,\sigma) = \sA(\sigma)X(t,\sigma)
$.
Its solution
$
    X(t,\sigma) = \re^{t\sA(\sigma)} X(0,\sigma)
$,
considered for different values of $\sigma \in \mT^\nu$,
admits a direct analogy (formulated in system theoretic terms) with the phonon theory \cite{S_1990} of collective oscillations  in
spatially periodic arrangements of atoms in
crystalline solids. More precisely, for a given $\sigma \in  \mT^\nu$, let $z \in \mC^n$ be an eigenvector of the matrix $\sA(\sigma)$ associated with its eigenvalue $s\in \mC$. 
Then the functions
\begin{equation}
\label{phon}
    x_k(t)
    =
    \Re(\re^{st + ik^{\rT} \sigma} z),
    \qquad
    k \in \mZ^\nu,\
    t\>0,
\end{equation}
satisfy the ODEs (\ref{xj}) with $u=0$. Indeed, substitution of the latter equality and (\ref{phon}) into the right-hand side of (\ref{xj}) yields
$
      \sum_{k \in \mZ^\nu}
  A_{j-k} x_k(t)
   =
  \Re
  \big(
  \re^{st + ij^{\rT} \sigma}   \sum_{k \in \mZ^\nu}\re^{-i(j-k)^{\rT} \sigma}A_{j-k}z
  \big)
  =
  \Re
  \big(
  \re^{st + ij^{\rT} \sigma}   \sA(\sigma)z
  \big)
  =
  \Re
  \big(
  s
  \re^{st + ij^{\rT} \sigma}   z
  \big) = \dot{x}_j(t)
$.
For purely imaginary eigenvalues $s = i\omega$ of the matrix $\sA(\sigma)$ (with $\omega \in \mR$), the solutions (\ref{phon}) describe persistent oscillations of frequency $\omega$ in the network, which are organised as (hyper) plane waves of length $\frac{2\pi}{|\sigma|}$ in $\mR^\nu$. Their wavefronts are orthogonal to the vector $\sigma$ and move (along $\sigma$) at constant phase velocity $-\frac{\omega}{|\sigma|}$, with the direction depending on the sign of $\omega$. Similarly to the dispersion relations of the phonon theory, the spectral structure of such oscillations  in the network is represented by the multi-valued map $\mT^\nu\ni \sigma\mapsto\omega \in \mR$ (of the wave vector to the frequency), which  will be discussed in Section~\ref{sec:phon}.

\section{SUPPLY RATE AND ENERGY BALANCE
}\label{sec:bal}

A generalised model for the work done by the input $u$ on the network is provided by a supply rate \cite{W_1972} at time $t\>0$ in the form
\begin{align}
\nonumber
    S(t)&:=
    \bra
        u(t),Gy(t)
    \ket
    =
    \sum_{j,k\in \mZ^\nu}
    u_j(t)^{\rT}
    G_{j-k}
    y_k(t)\\
\nonumber
    & =
    \frac{1}
    {(2\pi)^{\nu}}
    \int_{\mT^\nu}
    U(t,\sigma)^*
    \sG(\sigma)
    Y(t,\sigma)
    \rd \sigma\\
\label{S}
    & =
    \frac{1}
    {(2\pi)^{\nu}}
    \bra
        U(t,\cdot),
        \sG Y(t,\cdot)
    \ket,
\end{align}
where the Parseval identity is applied to the Fourier transforms $U$ and $Y$ from (\ref{uDFT}) and (\ref{yDFT}).
The quantity $S(t)$ is a bilinear function of the current input and output of the network
which is specified by a block Toeplitz matrix $G:= (G_{j-k})_{j,k\in \mZ^\nu}$, where $G_\ell \in \mR^{m\x r}$ are given matrices satisfying
\begin{equation}
\label{Ggood}
  \sum_{\ell \in \mZ^\nu}
  \|G_\ell\|^2 < +\infty,
\end{equation}
with the Fourier transform
\begin{equation}
\label{cG}
    \sG(\sigma)
    :=
    \sum_{\ell\in \mZ^\nu}
    \re^{-i\ell^{\rT}\sigma}
    G_\ell.
\end{equation}
It is assumed that the matrix $G$ describes a bounded linear operator from the output space $\fY$ in (\ref{fU_fX_fY}) to the input space $\fU$ in (\ref{fU_fX_fY}), so that
\begin{equation}
\label{Gnorm}
    \|G\|_{\infty}
    =
    \esssup_{\sigma \in \mT^\nu}
    \|\sG(\sigma)\|
    <+\infty,
\end{equation}
similarly to the operator norms in (\ref{Anorm})--(\ref{Dnorm}).
For example, the condition (\ref{Gnorm}) holds in the
case of a standard supply rate, when $u$ and $y$ consist of the corresponding force and velocity variables (with $r=m$) and $G$ is the identity operator. Returning to the general case, substitution of (\ref{Y}) into (\ref{S}) leads to
\begin{align}
\nonumber
    S(t)
     =&
    \frac{1}
    {(2\pi)^{\nu}}
    \int_{\mT^\nu}
    \Re
    \big(
    U(t,\sigma)^*
    \sG(\sigma)\\
\nonumber
    & \x
    (\sC(\sigma) X(t,\sigma)+\sD(\sigma) U(t,\sigma))
    \big)
    \rd \sigma\\
\nonumber
    = &
    \frac{1}
    {2(2\pi)^{\nu}}
    \int_{\mT^\nu}
    Z(t,\sigma)^*\\
\nonumber
    & \x
    {\begin{bmatrix}
        0 & \sC(\sigma)^* \sG(\sigma)^*\\
        \sG(\sigma) \sC(\sigma) & \sG(\sigma) \sD(\sigma) + \sD(\sigma)^* \sG(\sigma)^*
    \end{bmatrix}}\\
\label{W1}
    & \x
    Z(t,\sigma)
    \rd \sigma,
\end{align}
where
\begin{equation}
\label{Z}
    Z(t,\sigma)
    :=
      \begin{bmatrix}
        X(t,\sigma)\\
        U(t,\sigma)
    \end{bmatrix}
    =
    \sum_{k \in \mZ^\nu}
    \re^{-ik^{\rT}\sigma}
      \begin{bmatrix}
        x_k(t)\\
        u_k(t)
    \end{bmatrix}
\end{equation}
is the Fourier transform of the augmented state-input pair of the network at time $t\>0$.
Now, let $V:= (V_{j-k})_{j,k\in \mZ^\nu}$ be a symmetric block Toeplitz matrix, which is specified by a square summable matrix-valued function $\mZ^\nu \ni \ell\mapsto V_\ell \in \mR^{n\x n}$  on the lattice  satisfying $
    V_{-\ell} = V_\ell^{\rT}
$. Then the corresponding Fourier transform
\begin{equation}
\label{cV}
    \sV(\sigma)
    :=
    \sum_{\ell\in \mZ^\nu}
    \re^{-i\ell^{\rT}\sigma}
    V_\ell
\end{equation}
takes values in the subspace of complex Hermitian matrices of order $n$ and satisfies
\begin{equation}
\label{Vsymm}
  \sV(-\sigma)
  =
  \sV(\sigma)^{\rT}
  =
  \overline{\sV(\sigma)},
  \qquad
  \sigma \in \mT^\nu.
\end{equation}
Under the additional condition
\begin{equation}
\label{Vnorm}
    \|V\|_{\infty}
    =
    \esssup_{\sigma \in \mT^\nu}
    \|\sV(\sigma)\|
    <+\infty,
\end{equation}
the matrix $V$ describes a self-adjoint operator on the Hilbert space $\fX$ in (\ref{fU_fX_fY}). This gives rise to a quadratic form of the network state at time $t\>0$:
\begin{align}
\nonumber
    H(t)
    & :=
    \frac{1}{2}
    \bra
        x(t),
        Vx(t)
    \ket
    =
    \frac{1}{2}
    \sum_{j,k\in \mZ^\nu}
    x_j(t)^{\rT}
    V_{j-k}
    x_k(t)\\
\label{H}
    & = \frac{1}{(2\pi)^\nu}
    \int_{\mT^\nu}
    \sH(t,\sigma)
    \rd \sigma,
\end{align}
where
\begin{equation}
\label{Hsigma}
    \sH(t,\sigma)
    :=
    \frac{1}{2}
      X(t,\sigma)^*
    \sV(\sigma)
    X(t,\sigma).
\end{equation}
Here, the Parseval identity is applied to the Fourier transform $X$ from  (\ref{xDFT}), whose contribution to $H(t)$ at a given $\sigma \in \mT^\nu$ is quantified by the real-valued quantity $\sH(t,\sigma)$ in (\ref{Hsigma}).
The spatial frequency domain representation (\ref{X}) of the network state dynamics allows the time derivative of $H$ in (\ref{H}) to be computed as
\begin{equation}
\label{Hdot}
    \dot{H}(t)
    =
    \frac{1}{(2\pi)^\nu}
    \int_{\mT^\nu}
    \d_t \sH(t,\sigma)
    \rd \sigma,
\end{equation}
with
\begin{align}
\nonumber
    \d_t &\sH(t,\sigma)
    =
    \Re
    \big(
    X(t,\sigma)^*
    \sV(\sigma)
    \d_t X(t,\sigma)
    \big)    \\
\nonumber
    &=
    \Re
    \big(
    X(t,\sigma)^*
    \sV(\sigma)
    (\sA(\sigma) X(t,\sigma) + \sB(\sigma) U(t,\sigma))
    \big)    \\
\label{Hsigmadot}
    &=
    \frac{1}{2}
    Z(t,\sigma)^*
    {\small\begin{bmatrix}
        \sA(\sigma)^* \sV(\sigma) + \sV(\sigma) \sA(\sigma) & \sV(\sigma) \sB(\sigma)\\
        \sB(\sigma)^* \sV(\sigma) & 0
    \end{bmatrix}}
    Z(t,\sigma),
\end{align}
where use is made of (\ref{Z}).
If $H(t)$ in (\ref{H}) describes the internal energy (usually   referred to as the Hamiltonian) of the network at the current moment of time $t\>0$, then the difference
\begin{equation}
\label{diff}
    S(t) - \dot{H}(t)
    =
    \frac{1}
    {2(2\pi)^{\nu}}
    \int_{\mT^\nu}
    Z(t,\sigma)^*
    N(\sigma)
    Z(t,\sigma)
    \rd \sigma
\end{equation}
can be interpreted as the energy dissipation rate. This quantity is part of the
supply rate in (\ref{S}) which is not accounted for by the rate of change of the internal energy in (\ref{Hdot}).  The complex Hermitian matrix
\begin{align}
\nonumber
    N(\sigma)
    := &
      \left[{\begin{matrix}
        -\sA(\sigma)^* \sV(\sigma) - \sV(\sigma) \sA(\sigma)\\
        \sG(\sigma) \sC(\sigma)-\sB(\sigma)^* \sV(\sigma)
    \end{matrix}}\right.\\
\label{N}
      & \qquad\qquad\qquad
      \left.{\begin{matrix}
        \sC(\sigma)^* \sG(\sigma)^*-\sV(\sigma) \sB(\sigma)\\
        \sG(\sigma) \sD(\sigma) + \sD(\sigma)^* \sG(\sigma)^*
    \end{matrix}}
    \right]
\end{align}
of order $n+m$ in (\ref{diff}) is obtained by using (\ref{W1}), (\ref{Hdot}) and (\ref{Hsigmadot}),
and satisfies
\begin{equation}
\label{Nsymm}
  N(-\sigma)
  =
  N(\sigma)^{\rT}
  =
  \overline{N(\sigma)},
  \qquad
  \sigma \in \mT^\nu,
\end{equation}
similarly to (\ref{Vsymm}).
In order for the network model to correspond to a real physical system, the dissipation rate in (\ref{diff}) has to be nonnegative in view of the total energy conservation, and hence,
\begin{equation}
\label{diss}
    S(t) \> \dot{H}(t).
\end{equation}
A sufficient condition for this inequality to hold  for arbitrary admissible inputs $u$ and  initial network states $x(0)$ is positive semi-definiteness of the matrix (\ref{N}):
\begin{equation}
\label{Npos}
  N(\sigma) \succcurlyeq 0,
  \qquad
  {\rm almost\ all}\
  \sigma \in \mT^\nu.
\end{equation}
This condition is also necessary under an additional controllability assumption.
\begin{thm}
\label{th:Npos}
Suppose the matrix pair $(\sA(\sigma), \sB(\sigma))$ for the network in (\ref{xj}) and (\ref{yj}), given by (\ref{ABCD}), is controllable for any $\sigma\in \mT^\nu$. Then the dissipation rate in (\ref{diff}) is nonnegative for any square summable initial network states and locally integrable inputs in (\ref{ugood}) if and only if the matrix (\ref{N}) satisfies (\ref{Npos}).\hfill$\square$
\end{thm}
\begin{proof}
In view of (\ref{diff}), the fulfillment of (\ref{Npos}) implies (\ref{diss}) regardless of the controllability of $(\sA,\sB)$, thus proving the sufficiency of (\ref{Npos}), mentioned above. In order to show the necessity, let $\zeta: \mT^\nu\to \mC^{n+m}$ be a measurable  function which is Hermitian (that is, $\overline{\zeta(\sigma)} = \zeta(-\sigma)$ for any $\sigma\in \mT^\nu$) and such that $\zeta(\sigma)$ is a unit eigenvector
associated with the smallest eigenvalue
\begin{equation}
\label{Nmin}
    \lambda(\sigma)
    :=
    \lambda_{\min}(N(\sigma))
\end{equation}
of the Hermitian  matrix $N(\sigma)$ in (\ref{N}), so that, in accordance with (\ref{Nsymm}),
\begin{equation}
\label{zeta}
  N(\sigma)\zeta(\sigma) = \lambda(\sigma)\zeta(\sigma).
\end{equation}
Note that, due to (\ref{Nsymm}), the function $\lambda$ in (\ref{Nmin}) is symmetric.
Therefore, if the above controllability condition is satisfied, then for any time horizon $t>0$, there exists an admissible  network input $u$ on the time interval $[0,t]$ such that the vector (\ref{Z}) satisfies
\begin{equation}
\label{Zmin}
    Z(t,\sigma)
    =
    \left\{
    \begin{matrix}
    \zeta(\sigma) & {\rm if}\ \lambda(\sigma)<0\\
    0 & {\rm otherwise}
    \end{matrix}
    \right.
\end{equation}
for any $\sigma \in \mT^\nu$.
For such an input, in view of (\ref{Nmin})--(\ref{Zmin}), the energy dissipation rate (\ref{diff}) takes the form
\begin{equation}
\label{diffmin}
    S(t) - \dot{H}(t)
    =
    \frac{1}
    {2(2\pi)^{\nu}}
    \int_{\mT^\nu}
    \min(\lambda(\sigma),\, 0)
    \rd \sigma.
\end{equation}
If the property (\ref{Npos}) does not hold, then the set $\{\sigma \in \mT^\nu:\ \lambda(\sigma)<0\}$ is of positive Lebesgue measure, and the right-hand side of (\ref{diffmin})  is negative, thus contradicting (\ref{diss}). This proves that, under the controllability assumption on the pair $(\sA,\sB)$, the condition (\ref{Npos}) is necessary for the network dissipativity.
%
%
\end{proof}

If the self-adjoint operator $V$ is positive semi-definite (which is equivalent to that the corresponding Fourier transform $\sV$ in (\ref{cV}) satisfies $\sV(\sigma) \succcurlyeq 0$ for almost all $\sigma \in \mT^\nu$), and the network is initialized at zero state $x(0)=0$, then the integration of both parts of (\ref{diss}) over a time interval $[0,T]$ leads to
\begin{equation}
\label{Wintpos}
  \int_0^T S(t)\rd t
  \>
  H(T)-H(0) = H(T)\> 0
\end{equation}
for any $T>0$.
These relations hold for any network input $u$ such that $\|u(t)\|_2$ is a locally square integrable function of time $t\>0$ in the sense that
\begin{equation}
\label{ugoodT}
    \int_0^T
    \|u(t)\|_2^2
    \rd t
    =
    \int_0^T
    \sum_{k\in \mZ^\nu}
    |u_k(t)|^2
    \rd t
    <+\infty,
    \qquad
    T>0.
\end{equation}
The condition (\ref{ugoodT}) is stronger than (\ref{ugood}) and guarantees finiteness of the work by such an input on the network over any bounded  time interval. More precisely, (\ref{S}) implies that
\begin{align}
\nonumber
  \Big(\int_0^T
  |S(t)|
  \rd t
  \Big)^2
  & \<
  \int_0^T
  \|u(t)\|_2^2\rd t
  \int_0^T
  \|Gy(t)\|_2^2\rd t  \\
\label{Wgood}
  & \<
  \|G\|_{\infty}^2
  \int_0^T
  \|u(t)\|_2^2\rd t
  \int_0^T
  \|y(t)\|_2^2\rd t
\end{align}
in view of the Cauchy-Bunyakovsky-Schwarz inequality and the boundedness (\ref{Gnorm}) of the operator $G$. Here, the time-local square integrability of the input allows (\ref{ygood}) to be enhanced as
\begin{align}
\nonumber
    \frac{1}{2}
  \int_0^T&
  \|y(t)\|_2^2\rd t
  \<
    \int_0^T
    (
    \|Cx(t)\|_2^2 +
    \|Du(t)\|_2^2
    )
    \rd t\\
\label{ygood1}
  \< &
    \|C\|_{\infty}^2
    \int_0^T
    \|x(t)\|_2^2
    \rd t
    +
    \|D\|_{\infty}^2
    \int_0^T
    \|u(t)\|_2^2
    \rd t,
\end{align}
which is obtained by applying the inequality $\frac{1}{2}\|\alpha + \beta\|^2 \< \|\alpha\|^2 + \|\beta\|^2$ in an arbitrary Hilbert space to the right-hand side of (\ref{y}) and using the boundedness (\ref{Cnorm}) and (\ref{Dnorm}) of the operators $C$ and $D$. Also, (\ref{xgood}) implies that
 \begin{align}
\nonumber
    \|x(t)\|_2^2
    & \<
    2
    \re^{2\|A\|_{\infty}t}
    \Big(
        \|x(0)\|_2^2
        +
    \|B\|_{\infty}^2
        \Big(
    \int_0^t
    \|u(\tau)\|_2
    \rd \tau
    \Big)^2
    \Big)\\
\label{xnormup}
    & \<
    2
    \re^{2\|A\|_{\infty}t}
    \Big(
        \|x(0)\|_2^2
        +
        t
    \|B\|_{\infty}^2
    \int_0^t
    \|u(\tau)\|_2^2
    \rd \tau
    \Big),
\end{align}
where the Cauchy-Bunyakovsky-Schwarz inequality is used again. The local square integrability of $\|x(t)\|_2$ as a function of time $t\>0$ follows from (\ref{xnormup})  in view of (\ref{x0norm}) and (\ref{ugoodT}). In combination with (\ref{ygood1}), this ensures the time-local square integrability of the network output: $  \int_0^T
  \|y(t)\|_2^2\rd t <+\infty$ for any $T>0$. From the last property, (\ref{ugoodT}) and (\ref{Wgood}), it follows that the supply rate $S(t)$ is indeed locally integrable with respect to time.

\section{NETWORK PASSIVITY CONDITIONS IN THE SPATIAL FREQUENCY DOMAIN}\label{sec:pass}

The relations (\ref{Wintpos}) imply that the network is passive 
in the sense that any time-locally square integrable input $u$ in (\ref{ugoodT}) performs a nonnegative work:
\begin{equation}
\label{W}
    W(T)
    :=
  \int_0^T S(t)\rd t
  \>
  0,
  \qquad
  T>0.
\end{equation}
The passivity
can also be considered irrespective of a specific internal energy (storage) function (provided the initial network state is zero). For what follows, the network input $u$ is assumed to be time-space square summable in the sense that
\begin{equation}
\label{ugood1}
    \sn u \sn
    :=
    \sqrt{
    \int_0^{+\infty}
    \|u(t)\|_2^2
    \rd t}
    =
    \sqrt{
    \int_0^{+\infty}
    \sum_{k \in \mZ^\nu}
    |u_k(t)|^2
    \rd t}
    <+\infty,
\end{equation}
where
the norm $\sn \cdot \sn$ is associated with the inner product
$
    \dbra f, g \dket
    :=
    \int_0^{+\infty}
    \bra
        f(t),
        g(t)
    \ket
    =
    \int_0^{+\infty}
    \sum_{k \in \mZ^\nu}
    f_k(t)^* g_k(t)
    \rd t
$
for real or complex vector-valued functions $f(t):= (f_k(t))_{k \in \mZ^\nu}$ and $g(t):=(g_k(t))_{k \in \mZ^\nu}$ of the time and space variables. By the Parseval identity, (\ref{ugood1}) yields
$    \int_{\mR \x \mT^{\nu}}
    |\sU(i\omega, \sigma)|^2
    \rd \omega \rd \sigma
    =
    (2\pi)^{\nu+1}
    \sn u \sn^2
$,
where the function $\sU$ is given by (\ref{cU}). Although (\ref{ugood1}) is a stronger condition than (\ref{ugoodT}), the input $u$ produces the same work $W(T)$ in (\ref{W}) over a given time interval $[0,T]$  as its ``truncated'' version
\begin{equation}
\label{PiT}
\Pi_T(u)(t)
:=
\left\{
\begin{matrix}
u(t) & {\rm if}\ 0\< t\< T\\
0    & {\rm otherwise}
\end{matrix}
\right.
\end{equation}
(with $\Pi_T$ being an orthogonal projection operator for any $T>0$).
This property follows from causality of the network as an input-output operator $u\mapsto y$ and from the fact that the   supply rate $S(t)$ in (\ref{S}) is a bilinear function of $u(t)$ and $y(t)$.
\begin{thm}
\label{th:pass}
Suppose the network, governed by (\ref{xj}) and (\ref{yj}), satisfies (\ref{ABCDgood}) and (\ref{Anorm})--(\ref{Dnorm}), and is endowed with the supply rate (\ref{S}) subject to (\ref{Ggood}) and (\ref{Gnorm}). Also, suppose the matrix $\sA(\sigma)$ in (\ref{ABCD}) is Hurwitz for any $\sigma \in \mT^\nu$. Then the network is passive in the sense of (\ref{W})  for zero initial states and arbitrary time-locally square integrable inputs $u$ in (\ref{ugoodT}) if and only if its
spatio-temporal transfer function (\ref{F}) satisfies
\begin{equation}
\label{FGpos}
  E(\omega, \sigma)
  :=
  F(i\omega,\sigma)^*\sG(\sigma)^*
  +
  \sG(\sigma)F(i\omega,\sigma)
  \succcurlyeq
  0
\end{equation}
for almost all $(\omega,\sigma) \in \mR\x \mT^\nu$. Here, the function $\sG$ is given by (\ref{cG}).
\hfill$\square$
\end{thm}
\begin{proof}
Assuming that the network is initialized at zero, consider the work $W(T)$ up until a given time horizon $T>0$. Since an arbitrary admissible input $u$ (in the sense of (\ref{ugoodT})) can be replaced with $\Pi_T(u)$ in (\ref{PiT}) without affecting the work, we will assume, without loss of generality, that $u$ vanishes beyond the time interval $[0,T]$, and hence, so also does the supply rate $S$ in (\ref{S}). For any such input, substitution of (\ref{S}) into (\ref{W}) leads to
\begin{align}
\nonumber
    W(&T)
    =
    \int_0^{+\infty}
    S(t)\rd t
    =
    \dbra
        u, Gy
    \dket\\
\nonumber
    = &
    \frac{1}
    {(2\pi)^{\nu+1}}
    \int_{\mT^\nu}
    \Big(
    \int_{-\infty}^{+\infty}
    \sU(i\omega,\sigma)^*
    \sG(\sigma)
    \sY(i\omega,\sigma)
    \rd \omega
    \Big)
    \rd \sigma\\
\nonumber
    = &
    \frac{1}
    {(2\pi)^{\nu+1}}
    \int_{\mR\x \mT^\nu}
    \sU(i\omega,\sigma)^*
    \sG(\sigma)
    F(i\omega,\sigma)
    \sU(i\omega,\sigma)
    \rd \omega
    \rd \sigma\\
\label{work}
    = &
    \frac{1}
    {2(2\pi)^{\nu+1}}
    \int_{\mR\x \mT^\nu}
    \sU(i\omega,\sigma)^*
    E(\omega,\sigma)
    \sU(i\omega,\sigma)
    \rd \omega
    \rd \sigma,
\end{align}
where the function $F$ in (\ref{F}) is well-defined on the set $(i\mR)\x \mT^\nu$ (with $i\mR$ the imaginary axis) due to the matrix $\sA(\sigma)$ being Hurwitz for all $\sigma \in \mT^\nu$.
Here, in view of  (\ref{sgood}), the Parseval identity is used together with the Fourier transforms (\ref{cY}), (\ref{cU}) and the Hermitian matrix $E(\omega,\sigma) \in \mC^{m\x m}$ in (\ref{FGpos}).
In view of (\ref{work}), the positive semi-definiteness of $E$ almost everywhere in $\mR\x \mT^{\nu}$ implies that $W(T)\> 0$ for all admissible inputs vanishing outside the time interval $[0,T]$. The necessity of the matrix inequality in (\ref{FGpos}) for the network passivity can be obtained by letting $T\to +\infty$ and considering all possible time-space square summable inputs satisfying (\ref{ugood1}).
%
%
\end{proof}

Note that in the case of $r=m$ and the standard supply rate mentioned above,  $\sG$ in (\ref{cG}) is the identity matrix, and (\ref{FGpos}) reduces to
\begin{equation}
\label{Fposreal}
  F(i\omega,\sigma) +
  F(i\omega,\sigma)^*
  \succcurlyeq
  0,
  \qquad
  \omega \in \mR,\
  \sigma \in \mT^\nu,
\end{equation}
which is a network counterpart of the positive real property. 
On the other hand, in an extended setting, the block Toeplitz matrix $G$ can be formed from differential operators $G_\ell(\d_t)$ with respect to time (whose entries  are, for example, polynomials of $\d_t$). In this case, Theorem~\ref{th:pass} is modified by replacing the function $\sG$ in (\ref{FGpos}), given by (\ref{cG}), with the Fourier-Laplace transform
\begin{equation}
\label{cGgen}
    \sG(s,\sigma)
    :=
    \sum_{\ell\in \mZ^\nu}
    \re^{-i\ell^{\rT}\sigma}
    G_\ell(s).
\end{equation}
In particular, if $r=m$ and the extended operator $G$ acts on the network output as $Gy = \dot{y}$, then (\ref{cGgen}) yields $\sG(s,\sigma) = sI_m$. In accordance with the structure  of the supply rate (\ref{S}),  this describes the setting when $u$ and $y$ consist of the corresponding force and position (rather than velocity) variables.  In this case,
 (\ref{FGpos})  takes the form
\begin{equation}
\label{Fnegimag0}
  i\omega (F(i\omega,\sigma) -
  F(i\omega,\sigma)^*)
  \succcurlyeq
  0,
  \qquad
  \omega \in \mR,\
  \sigma \in \mT^\nu.
\end{equation}
Since the spatio-temporal transfer function $F$ in (\ref{F}) satisfies $F(-i\omega, -\sigma) = \overline{F(i\omega,\sigma)}$ for all $\omega \in \mR$, $\sigma \in \mR^\nu$, and the complex conjugation of a Hermitian matrix preserves positive semi-definiteness, then  (\ref{Fnegimag0}) is equivalent to
\begin{equation}
\label{Fnegimag}
    \frac{1}{i}
  (F(i\omega,\sigma) -
  F(i\omega,\sigma)^*)
  \preccurlyeq
  0,
  \qquad
  \omega \in \mR_+,\
  \sigma \in \mT^\nu.
\end{equation}
Similarly to (\ref{Fposreal}), the condition (\ref{Fnegimag}) is a network version of the negative-imaginary property \cite{PL_2010,XPL_2010}.
%
%
%
%
In the case when the matrix $\sA(\sigma)$ in (\ref{ABCD}) is Hurwitz for all $\sigma \in \mT^\nu$ (as assumed in Theorem~\ref{th:pass}), there exists a positive definite Hermitian $\mC^{n\x n}$-valued function $\sV$ in (\ref{cV}) on the torus $\mT^\nu$ satisfying (\ref{Vsymm}), (\ref{Vnorm}) and such that
\begin{equation}
\label{ALI}
    \sA(\sigma)^* \sV(\sigma) +\sV(\sigma) \sA(\sigma) \prec 0,
    \qquad
    \sigma \in \mT^\nu.
\end{equation}
In isolation from the environment (when $u=0$), the quadratic stability of the network can be formulated by enhancing the positive definiteness as
\begin{equation}
\label{Vgood}
    \mu
    :=
    \essinf_{\sigma \in \mT^\nu}
    \lambda_{\min}(\sV(\sigma))
    >0.
\end{equation}
Then
the corresponding block Toeplitz matrix $V$ describes a positive definite self-adjoint operator on $L^2(\mZ^\nu, \mR^n)$ whose inverse $V^{-1}$ is also     such an operator and its norm is related to the quantity $\mu$ in (\ref{Vgood}) by
\begin{equation}
\label{Vinvnorm}
    \|V^{-1}\|_{\infty}
    =
    \frac{1}{\mu}.
\end{equation}
Furthermore, the Hamiltonian $H$, associated with $V$ by (\ref{H}), admits the following bounds  in terms of the standard $L^2$-norm of the network state:
\begin{equation}
\label{Hbounds}
    \frac{\mu}{2}
    \|x(t)\|_2^2
    \<
    H(t)
    \<
    \frac{\|V\|_{\infty}}{2}
    \|x(t)\|_2^2.
\end{equation}
Similarly to finite-dimensional settings,
the dissipativity (\ref{diss}) (or a nonstrict version of the inequality  (\ref{ALI})) implies that the Hamiltonian of the isolated network does not increase with time (that is, $\dot{H}\< 0$), which, in combination with (\ref{Hbounds}) leads to
\begin{equation}
\label{cond}
    \|x(t)\|_2^2
    \<
    \frac{2}{\mu} H(t)
    \<
    \frac{2}{\mu} H(0)
      \< \frac{\|V\|_\infty}{\mu} \|x(0)\|_2^2,
    \quad
    t\>0
\end{equation}
(similar bounds, arising from quadratic Hamiltonians, are also used for quantum harmonic oscillators \cite[Eq. (22)]{P_2014}).
In view of (\ref{Vinvnorm}), the factor
$
    \frac{\|V\|_\infty}{\mu} = \|V\|_{\infty}\|V^{-1}\|_{\infty}
    \> 1
$
on the right-hand side of (\ref{cond}) is the condition number of the operator $V$ (which quantifies how far $V$ is from scalar operators).

\section{PHONON THEORETIC
DISPERSION RELATIONS}\label{sec:phon}

We will now consider a class of translation invariant networks, which, in isolation from the environment, manifest phonon-like dynamics mentioned in Section~\ref{sec:freq}. Suppose the state vectors  $x_k$  of the component systems consist of the conjugate position and momentum variables which are assembled  into $\frac{n}{2}$-dimensional vectors $q_k$ and $p_k$, respectively (with $n$ being even):
\begin{equation}
\label{xqp}
    x_k
    =
    \begin{bmatrix}
        q_k\\
        p_k
    \end{bmatrix},
    \qquad
    k \in \mZ^\nu.
\end{equation}
Accordingly, the network is assumed to have the following Hamiltonian:
\begin{equation}
\label{Hqp}
    H
    :=
    \frac{1}{2}
    \sum_{j\in \mZ^\nu}
    \Big(
        \|p_j\|_{M^{-1}}^2
        +
        q_j^{\rT}
        \sum_{k\in \mZ^\nu}
        K_{j-k}
        q_k
    \Big).
\end{equation}
Here, $M$ is a real positive definite symmetric mass matrix of order $\frac{n}{2}$, and use is made of a weighted Euclidean norm $\|v\|_L := \sqrt{v^{\rT} L v} = |\sqrt{L}v|$ of a vector $v$ specified by such a matrix $L$. Also,
$K:=(K_{j-k})_{j,k\in \mZ^\nu}$ is a symmetric block Toeplitz stiffness operator whose blocks $K_{\ell} \in \mR^{\frac{n}{2}\x \frac{n}{2}}$ specify the potential energy part of the Hamiltonian. In view of (\ref{xqp}) and (\ref{Hqp}),  the blocks of the corresponding
block Toeplitz matrix $V$ in (\ref{H}) are given by
\begin{equation}
\label{VKM}
    V_{\ell}
    =
    \begin{bmatrix}
        K_\ell & 0\\
        0 & \delta_{0\ell}M^{-1}
    \end{bmatrix},
    \qquad
    \ell \in \mZ^\nu,
\end{equation}
where
$\delta_{jk}$ is the Kronecker delta. In application to the isolated network, the Hamiltonian equations of motion take the form
\begin{equation}
\label{qpdot}
    \dot{x}_j
    =
    \begin{bmatrix}
        \dot{q}_j\\
        \dot{p}_j
    \end{bmatrix}
    =
    J \d_{x_j}H
    =
    \begin{bmatrix}
        \d_{p_j}H\\
        -\d_{q_j}H
    \end{bmatrix}
    =
    \begin{bmatrix}
        M^{-1} p_j\\
        -\sum_{k\in \mZ^\nu}
    K_{j-k} q_k
    \end{bmatrix}
\end{equation}
for all $j\in \mZ^\nu$,
where
\begin{equation}
\label{J}
    J:=
    \begin{bmatrix}
        0 & 1 \\
        -1 & 0
    \end{bmatrix}
    \ox I_{n/2}
\end{equation}
is the symplectic structure matrix associated with the partitioning (\ref{xqp}), with $\ox$ the Kronecker product of matrices. A comparison of (\ref{qpdot}) with (\ref{xj}) allows the blocks of the matrix $A$ in (\ref{x}) to be expressed in terms of (\ref{VKM}) and (\ref{J}) as
\begin{equation}
\label{AKM}
  A_{\ell}
  =
  J V_{\ell}
  =
    \begin{bmatrix}
        0 & \delta_{0\ell}M^{-1}\\
        -K_\ell & 0
    \end{bmatrix},
    \qquad
    \ell \in \mZ^\nu.
\end{equation}
Upon substituting (\ref{AKM}) into (\ref{ABCD}),
the function $\sA$ takes the form
\begin{equation}
\label{sAKM}
  \sA(\sigma)
  =
    \begin{bmatrix}
        0 & M^{-1}\\
        -\sK(\sigma) & 0
    \end{bmatrix},
\end{equation}
where
\begin{equation}
\label{sK}
    \sK(\sigma)
    :=
    \sum_{\ell\in \mZ^\nu}
    \re^{-i\ell^{\rT}\sigma}
    K_\ell.
\end{equation}
Since $M\succ 0$, then $s\in \mC$ is an eigenvalue of the matrix $\sA(\sigma)$ in (\ref{sAKM}) if and only if $s^2$ is an eigenvalue of $-M^{-1/2}\sK(\sigma)M^{-1/2}$. Therefore, if the stiffness operator $K$ is positive semi-definite (so that $0$ delivers a minimum to the potential energy), then $\sA(\sigma)$ has a purely imaginary spectrum:
\begin{equation}
\label{spec}
    \fS(\sA(\sigma))
    =
    \big\{
        \pm i\sqrt{\lambda}:\
        \lambda \in \fS(\sK(\sigma)M^{-1})
    \big\}
\end{equation}
for any $\sigma \in \mT^\nu$. Such a network is organised as an infinite-dimensional spring-mass system whose vibrations are linear superpositions of ``basis'' phonons given by (\ref{phon}) with $s=\pm i \sqrt{\lambda}$ for different values of $\sigma \in \mT^\nu$. The matrix $\sK(\sigma)M^{-1} = \sqrt{M}M^{-1/2}\sK(\sigma)M^{-1/2}M^{-1/2}$ in (\ref{spec}),  which is isospectral to $M^{-1/2}\sK(\sigma)M^{-1/2}\succcurlyeq 0$, corresponds to the stiffness-to-mass ratio for one-mode harmonic oscillators. The dynamic properties of phonons with large wavelengths $\frac{2\pi}{|\sigma|}$ (manifesting themselves in highly correlated motions of distant subsystems in the network) depend on the asymptotic behaviour of the matrix $\sK(\sigma)$ in (\ref{sK}) in a small neighbourhood of $\sigma = 0$.


\begin{thm}
\label{th:phon}
Suppose the isolated translation invariant network (\ref{qpdot}) with the Hamiltonian (\ref{Hqp}) has a positive semi-definite stiffness operator $K$ satisfying
\begin{align}
\label{K0}
    \sum_{\ell \in \mZ^\nu}
    K_{\ell} & = 0,\\
\label{K2}
  \sum_{\ell \in \mZ^\nu}
  |\ell|^2
  \|K_\ell\|
  & <+\infty.
\end{align}
Then the phase velocities of the phonons in the network are uniformly bounded:
\begin{equation}
\label{speed}
  \sup
  \Big\{
        \frac{\omega}{|\sigma|}:\
        \sigma
        \in
        \mT^\nu\setminus\{0\},\,
        \omega
        \in
        \mR_+,\,
        i\omega
        \in
        \fS(\sA(\sigma))
  \Big\}
  <+\infty.
\end{equation}
\hfill$\square$
\end{thm}
\begin{proof}
The condition (\ref{K2}) ensures twice continuous differentiability of the function $\sK$ in (\ref{sK}).
Since $K \succcurlyeq 0$ is equivalent to $\sK(\sigma)\succcurlyeq 0$ for all $\sigma \in \mT^\nu$, then (\ref{K0}) implies that $\sK(0)=0$. Hence, $\sK$ has zero Frechet derivative at the origin: $\sK'(0) = 0$ (that is, $\sum_{\ell \in \mZ^\nu}\ell \ox K_{\ell} = 0$). Indeed, otherwise, the leading term of the linearly truncated Taylor series expansion $\sK(\sigma) = \sum_{j=1}^\nu \sigma_j \d_{\sigma_j}\sK(\sigma)\big|_{\sigma = 0} +o(|\sigma|)$ would fail  to be a positive semi-definite matrix for arbitrary $\sigma:= (\sigma_k)_{1\< k\<\nu}$. Therefore, the quadratically truncated Taylor series expansion of $\sK$ in a small neighbourhood of the origin reduces to
\begin{align}
\nonumber
    \sK(\sigma)
    & =
    \frac{1}{2}
    \sum_{j,k=1}^\nu
    \sigma_j \sigma_k
    \Gamma_{jk}
    +
    o(|\sigma|^2)\\
\label{sKquad}
    & =
    |\sigma|^2
    \left(
    \frac{1}{2}
    \sum_{j,k=1}^\nu
    \theta_j \theta_k
    \Gamma_{jk}
    +
    o(1)
    \right),
    \quad
    {\rm as}\ \sigma \to 0.
\end{align}
Here, we have associated with $\sigma\ne 0$ the unit direction vector     $\theta
    :=
    (\theta_k)_{1\< k\< \nu}$ by
\begin{equation}
\label{theta}
    \theta
    :=
    \frac{\sigma}{|\sigma|},
\end{equation}
and the matrices $\Gamma_{jk} \in \mR^{\frac{n}{2}\x \frac{n}{2}}$ specify the second Frechet derivative of the function $\sK$  as
\begin{equation}
\label{Gammajk}
    \Gamma_{jk}
    :=
    \d_{\sigma_j}\d_{\sigma_k}
    \sK(\sigma)
    \big|_{\sigma = 0}
    =
    -
    \sum_{\ell \in \mZ^\nu}
    \ell_j\ell_k
    K_{\ell}
\end{equation}
in view of (\ref{sK}). The leading term     $\sum_{j,k=1}^\nu
    \sigma_j \sigma_k
    \Gamma_{jk}$ in (\ref{sKquad}) is a real positive semi-definite symmetric matrix of order $\frac{n}{2}$ for any $\sigma \in \mR^\nu$. Now, since the spectrum (as a set-valued function) of a matrix depends continuously on it \cite{H_2008,HJ_2007}, then, in view of (\ref{spec}),
\begin{equation}
\label{ommax}
  \sup_{\sigma \in \mT^\nu}
  \br(\sA(\sigma))
  =
  \sqrt{\max_{\sigma \in \mT^\nu}
  \lambda_{\max}(\sK(\sigma)M^{-1})}
  <+\infty,
\end{equation}
where $\lambda_{\max}(\cdot)$ is the largest eigenvalue of a matrix with a real spectrum. Hence, the left-hand side of (\ref{speed}) can be unbounded only because of the asymptotic behaviour of $\br(\sA(\sigma))$ as $\sigma \to 0$  (which corresponds to long-wave phonons). However, it follows from (\ref{spec}) and (\ref{sKquad})--(\ref{Gammajk}) that
\begin{align}
\nonumber
  &\limsup_{\sigma\to 0}
  \frac{\br(\sA(\sigma))}{|\sigma|}
  =
  \limsup_{\sigma\to 0}
  \frac{\sqrt{\lambda_{\max}(\sK(\sigma)M^{-1})}}{|\sigma|}  \\
\label{speed0}
  & =
  \max_{\theta\in \mR^\nu:\ |\theta| = 1}
  \sqrt{
  \frac{1}{2}
  \lambda_{\max}
  \Big(
      \sum_{j,k=1}^\nu
    \theta_j \theta_k
    \Gamma_{jk}
    M^{-1}
    \Big)
  }
   <+\infty.
\end{align}
A combination of (\ref{ommax}) and (\ref{speed0}) leads to
(\ref{speed}).
\end{proof}

The speed of propagation of disturbances is described in physics literature  in terms of their group velocity $\d_{\sigma}\omega$ (see, for example, \cite{LL_1980}). Note that the differentiability  is ensured for those phonons whose frequency $\omega$ corresponds to a non-degenerate  eigenvalue (of multiplicity one) of the matrix $\sK(\sigma)M^{-1}$ in (\ref{spec}).
%
%
%
%
%
We will now provide an example which illustrates the condition (\ref{K0}).
Consider
the movement of a thin isotropic plate of density $\rho>0$ and bending stiffness $\beta>0$, governed by the Kirchhoff-Love equation \cite{R_2007}:
\begin{equation}
\label{PDE}
    \rho\d_t^2 w = -\frac{\beta}{2} \Delta^2 w +f.
\end{equation}
Here, $w(t,\xi)$ is the local deviation of the plate from the equilibrium (unbent) position, which is a real-valued  function of time $t$ and the two-dimensional vector $\xi:= {\small\begin{bmatrix}\xi_1\\ \xi_2\end{bmatrix}} \in \mR^2$ of Cartesian coordinates on the plane, and $\Delta:= \d_{\xi_1}^2 + \d_{\xi_2}^2$ is the Laplacian over the spatial variables. Also, $f(t,\xi)$ is the density of the external force acting transversally on the plate. A finite-difference scheme  for the numerical solution of the PDE (\ref{PDE}) with spatial step size $h>0$ can be viewed as a translation invariant network on the two-dimensional integer lattice $\mZ^2$. By denoting $\mho:= (\mho_{jk})_{j,k\in \mZ}$ the spatial discretization of the solution of (\ref{PDE}) at time $t\>0$ (so that $\mho_{jk}(t)$ approximates $w(t,jh,kh)$), the corresponding functions
\begin{equation}
\label{qpmho}
    q:= \mho,
    \qquad
    p:= \rho \d_t \mho
\end{equation}
on the set $\mR_+\x \mZ^2$
are evolved according to the Hamilton equations
\begin{equation}
\label{qpmhodot}
    \dot{q}
    =
    \frac{1}{\rho} p,
    \qquad
    \dot{p}
    =
    -\frac{\beta}{2}L^2 q + g.
\end{equation}
Here, $g:= (f(t,jh,kh))_{j,k\in \mZ}$ is a discretization of the forcing term,
and $L$ is a self-adjoint Toeplitz operator which approximates the Laplacian $\Delta$ as
\begin{equation}
\label{Lmho}
    (L\mho)_{jk}
    :=
    \frac{1}{h^2}
    (
        q_{j+1,k}+q_{j-1,k}
        +
        q_{j,k+1}+q_{j,k-1}
        -
        4q_{jk}
    ).
\end{equation}
A comparison of (\ref{qpmho}) and (\ref{qpmhodot}) with (\ref{qpdot}) shows that the scalar $\rho$ plays the role of the mass matrix $M$, and $\frac{\beta}{2}L^2$ is the stiffness operator $K$. The latter implies that the Fourier transforms (\ref{sK}) and $
    \sL(\sigma)
    :=
    \sum_{\ell\in \mZ^2}
    \re^{-i\ell^{\rT}\sigma}
    L_\ell
$ are related by
$\sK = \frac{\beta}{2}\sL^2$, and hence, $K$ inherits the property (\ref{K0}) from the discretized Laplacian $L$ in (\ref{Lmho}), since $\sL(0) = \sum_{\ell \in \mZ^2} L_\ell = 0$.

\section{CONCLUSION}
\label{sec:conc}

We have considered a class of translation invariant networks of interacting linear systems at sites of a multidimensional lattice.
Energy-based input-output properties of such a network (including passivity, positive realness and the negative imaginary property)  have been considered in terms of its transfer function representation in the spatio-temporal frequency domain. We have also discussed quadratic stability of the network in the isolated regime, 
and dispersion relations for phonons in the isolated network. The results of the paper are applicable to quadratic regulation problems for large flexible structures with a sparse collocation of sensors and actuators.

\end{document}